\documentclass{amsart}

\title{Targeted Least Cardinality Candidate Key for Relational Databases}

\author{Vasileios Nakos}
\address{University of Athens 
\& RelationalAI, Inc.}
\email{vasilisnak@di.uoa.gr, vasileios.nakos@relational.ai}

\author{Hung Q. Ngo}
\address{RelationalAI, Inc.}
\email{hung.ngo@relational.ai}

\author{Charalampos E. Tsourakakis}
\address{RelationalAI, Inc.}
\email{charalampos.tsourakakis@relational.ai}

\usepackage{amsthm}
\usepackage{amssymb}
\usepackage{url}
\usepackage[...]{youngtab}
\usepackage{graphicx}
\usepackage{enumitem}
\usepackage{empheq}
\usepackage{xcolor}
\usepackage{hyperref}

\usepackage{multirow}
\usepackage[most]{tcolorbox}

\usepackage{algorithmic,algorithm}

\usepackage{xspace}
\usepackage{todonotes}

\newtheorem{problem}{Problem}
\newtheorem{conjecture}{Conjecture}
 \newtheorem{corollary}{Corollary}
 \newtheorem{theorem}{Theorem}
 \newtheorem{lemma}{Lemma}
\newtheorem{definition}{Definition}

\newcommand\myeq{\mathrel{\stackrel{\makebox[0pt]{\mbox{\normalfont\tiny def}}}{=}}}
\newcommand{\spara}[1]{\smallskip{\bf #1}}

\newcommand{\Prob}[1]{\ensuremath{{\bf{Pr}}\left[{#1}\right]}}
\newcommand{\Mean}[1]{\ensuremath{{\mathbb E}\left[{#1}\right]}}

\newcommand{\tcand}{TCAND\xspace}

\newcommand{\fd}{{\it FD}\xspace}
\newcommand{\fds}{{\it FDs}\xspace}

\newcommand{\setcover}{\textsc{Set Cover}
}

\newcommand{\rbsetcover}{\textsc{Red Blue Set Cover}
}

\begin{document}
\maketitle
 
\begin{abstract}

Functional dependencies (\fds) are a central theme in databases, playing a major role in the design of database schemas and the optimization of queries~\cite{ramakrishnan2003database}. In this work, we introduce  the {\it targeted least cardinality candidate key problem} (\tcand). This problem is defined over a set of functional dependencies \(\mathcal{F}\) and a target variable set $T \subseteq V$, and it aims to find the smallest set $X \subseteq V$  such that the \fd  $X \to T$  can be derived from $\mathcal{F}$. The \tcand problem generalizes the well-known NP-hard problem of finding the least cardinality candidate key~\cite{lucchesi1978candidate}, which has been previously demonstrated to be at least as difficult as the set cover problem.

We present an integer programming (IP) formulation for the \tcand problem, analogous to a layered set cover problem. We analyze its linear programming (LP) relaxation from two perspectives: we propose two approximation algorithms and investigate the integrality gap. Our findings indicate that the approximation upper bounds for our algorithms are not significantly improvable through LP rounding, a notable distinction from the standard    \setcover problem. Additionally, we discover that a generalization of the \tcand problem is equivalent to a variant of the    \setcover problem, named \rbsetcover~\cite{carr1999red}, which cannot be approximated within a sub-polynomial factor in polynomial time under plausible conjectures~\cite{chlamtavc2023approximating}. Despite the extensive history surrounding the issue of identifying the least cardinality candidate key, our research contributes new theoretical insights, novel algorithms, and demonstrates that the general \tcand problem poses complexities beyond those encountered in the    \setcover problem.

\end{abstract}

\section{Introduction} 
\label{sec:intro}

Relational databases are a fundamental component of modern information systems, used to store and manage vast amounts of data across a wide range of industries and applications~\cite{ramakrishnan2003database}. However, designing an effective database schema can be a complex task, requiring careful consideration of factors such as data integrity, performance, and scalability.  One key tool in this process is the use of functional dependencies (\fds), which provide a way to describe the relationships between attributes in a database relation~\cite{fagin1977functional}. A functional dependency (\fd) is a statement that indicates that the value of one or more attributes uniquely determines the value of another attribute. For example, 
consider a relation  with attributes {\sc Student ID}, {\sc Student Name}, and {\sc Student Email}.
In this relation, each student id is associated with a unique  student name and email. This means that if we know a student's ID, we can determine their name and email. This relationship can be represented by the following  \fd: {\sc Student ID} $\to$ {\sc Student Name}, {\sc Student Email}. Alternatively, we express that the {\sc Student ID} variable {\it functionally determines} the name and email, meaning that each id uniquely corresponds to one student name and email in the relation.
A key is a set of one or more attributes that uniquely identifies a tuple (or record) within a relation. A key ensures that there are no duplicate records in the table and establishes a way to reference records for relational operations.  A {\it candidate key} is any key that can serve as the primary key, i.e., a  set of attributes that uniquely identifies a tuple within a relation.  Finding a candidate key in a schema requires finding a set of attributes that functionally determine all the rest~\cite{fagin1984theory}. It is important to note that it is often customary to select a candidate key with the least number of attributes as the primary key. This approach simplifies the database design and can enhance storage efficiency and query speed, particularly when the primary key plays a role in indexing, joins, and various database tasks~\cite{ramakrishnan2003database}. 
For example, in query optimization, fewer attributes in a key mean fewer columns to index and process during queries, which enhances the speed and efficiency of database operations. Furthermore, using the smallest possible key simplifies the design of the database schema, making it easier to understand and maintain. It also reduces the likelihood of errors in data management and makes integrity checks more efficient. This problem is known as the {\it minimum candidate key problem}~\cite{lipski1977two,lucchesi1978candidate}.
Furthermore, finding the Boyce–Codd normal form (BCNF) to eliminate redundancy is based on functional dependencies~\cite{ramakrishnan2003database}. Understanding functional dependencies is essential for database designers and administrators to create efficient and effective database schemas that accurately store and manage data. Functional dependency analysis is a major topic with a variety of applications beyond schema design query optimization, that additionally include cost estimation, order optimization, 
selectivity estimation, estimation of (intermediate) result sizes and data cleaning among others~\cite{bhargava1996efficient,bhargava1995simplification,chu2014ruleminer,paulley2001exploiting,ilyas2015trends}. The {\bf DISTINCT} clause appears frequently in SQL queries and frequently requires an expensive sort to remove duplicates. Functional dependency analysis can identify redundant  {\bf DISTINCT} clauses~\cite{paulley2010exploiting}, thus lowering significantly the execution cost of a query. 

  In this work  we revisit functional dependency analysis and make several key contributions to this longstanding line of work. 

 $\bullet$  {\bf Novel formulation:} We introduce   a well-motivated, novel formulation that generalizes the classic problem of finding the least cardinality key of a schema~\cite{lipski1977two,lucchesi1978candidate}. We refer to this problem as the {\it targeted least cardinality candidate key generation} problem, or \tcand for short.

\begin{tcolorbox}
\begin{problem}\label{prob:tcand}
{\bf [Targeted Least Cardinality Candidate Key (\tcand)]} 
\normalfont 
Given a set of functional dependencies $\mathcal{F}$  and a set of target variables $T\subseteq V$, the aim is to determine  a set of variables $X \subseteq V$ with the smallest cardinality, such that the closure $X^+$ fulfills the condition $T \subseteq X^+$. In other terms,  the functional dependency $X\to T$ is logically implied   from the set $\mathcal{F}$.
\end{problem}
\end{tcolorbox}

It should be noted that the \tcand problem not only extends a well-established classical issue but also serves as a pivotal component in contemporary management systems for knowledge graph databases. Our approach is inherently tied to semantic optimizations within the RelationalAI's engine, which integrates various constraints, such as semi-ring constraints, into the query optimization framework~\cite{abo2016faq,abo2022convergence}. For instance, \tcand is commonly utilized to determine the requisite number of variables in a bag (or the maximum fractional edge cover number across the bags) within a tree decomposition; further information can be found in the research by Abo Khamis, Ngo, and Rudra~\cite{abo2016faq}. Importantly, in over half of these cases, the target set is a strict subset of $V$, that is, $T \subset V$. This indicates that the majority of queries do not revolve around finding a primary key. In our database engine, specifically for the TPC-H benchmark~\cite{dreseler2020quantifying}, addressing Problem~\ref{prob:tcand} becomes necessary more than 4\,800 times.

\noindent 

$\bullet$ {\bf Hardness:} We study in depth the hardness of approximation of the \tcand formulation. Despite the long history and importance of the problem in its general form when $T=V$~\cite{lucchesi1978candidate}, surprisingly not much is known on the approximability of the problem. As we show the known approximation lower bound~\cite{akutsu1996approximating} is tight only for some special cases, while the general case is significantly harder.  We achieve this result by establishing a novel connection with the \rbsetcover that allows us to leverage recent progress to establish the a powerful hardness result assuming the Dense-vs-Random conjecture~\cite{chlamtavc2017minimizing}.

$\bullet$ {\bf Exact IP formulation:} We present an exact integer programming (IP) formulation that represents the \tcand problem as a layered set cover problem. Intuitively, each layer corresponding to a round of $\fd$ inference. With recent advancements in solver software, our formulation can be practical for real-world use for instances of moderate size. From a theory perspective, it serves as the basis for designing approximation algorithms. 

$\bullet$ {\bf Approximation algorithms:}  
We design two approximation algorithms based on the linear programming relaxation of our integer programming (IP) formulation. Both algorithms rely on solving a variant of the \tcand problem we introduce,  parameterized by the number of rounds of inference $D$.

\begin{tcolorbox}
\begin{problem}\normalfont
\label{prob:rtcand}\textbf{[$D$-round-\tcand]}
Given a set of functional dependencies $\mathcal{F}$, and a set of target variables $T\subseteq V$ we want to find a least cardinality set of variables $X \subseteq V$ whose closure $X^+$ includes $T$, i.e., $X^+ \supseteq T$, by performing at most $D$ rounds of \fd inference. 
\end{problem}
\end{tcolorbox}

 Our approach is based on approximating Problem~\ref{prob:rtcand} for $D=1$;  this gives a natural approximation algorithm for the $D$-round \tcand problem with an exponential dependence on $D$.  It is worth noting that Problem~\ref{prob:tcand} is a special case of Problem~\ref{prob:rtcand} by setting $D=n$ as we explain later in detail. We also show that our approximation guarantee is asymptotically tight by studying the integrality gap of the LP relaxation.  

$\bullet$ {\bf Equivalence with the \rbsetcover problem:} We discover an equivalence between the \tcand problem and  the \rbsetcover problem~\cite{carr1999red}, a variant of the    \setcover  problem.  This discovery holds dual significance for the \tcand problem. Firstly, it introduces an additional approximation algorithm developed by Chlamtavc et al.~\cite{chlamtavc2023approximating} for the \rbsetcover problem, and secondly, it establishes an inapproximability result.

\section{Preliminaries} 
\label{sec:related}

\spara{Functional dependencies.} A relational database schema $R$, represented as $R(A_1, \ldots, A_n)$, consists of a set of attributes or variables. An instance of $R$, denoted by $r(R)$, is a set of tuples, where each tuple is an element of the Cartesian product of the attributes' domains, i.e., $r(R) \subseteq \text{dom}(A_1) \times \ldots \times \text{dom}(A_n)$. We also use $V = \{A_1, \ldots, A_n\}$, or simply $V = [n]$, to denote the set of attributes/variables. Each element $t \in r(R)$ is referred to as a tuple.  Functional dependencies (\fds for short) are properties of the semantics of the attributes in $R$~\cite{fagin1984theory}. Specifically  an \fd  is a constraint between two sets of attributes $X,Y \subseteq V$. We say that $X$ functionally determines $Y$ and this means that if two tuples have the same value for all the attributes in $X$, they must also have the same values for all the attributes in $Y$. This is denoted as $X \to Y$.   We refer to an \fd $X \to Y$ as regular when $|Y|=1$. Any irregular \fd $X \to y_1y_2\ldots y_k$ is equivalent to the set of regular \fds $X\to y_i$ for $i=1,\ldots,k$.  An input set  $\mathcal{F}$ of \fds, may logically imply more \fds. For instance, the set $\mathcal{F}=\{ a \to b, b \to c\}$ logically implies $a \to c$. The set of all possible valid \fds for $\mathcal{F}$ is its closure $\mathcal{F}^+$.  The inference of valid \fds is performed using Armstrong's axioms which are sound and complete~\cite{fagin1977functional}.  In practice $|\mathcal{F}| \ll |\mathcal{F}^+|$.   Finding a smaller set of \fds $\mathcal{F}'$ than $\mathcal{F}$ such that $\mathcal{F}'^+=\mathcal{F}^+$ is known as the canonical cover problem and can be solved efficiently~\cite{ramakrishnan2003database}. Other compression schemes are also available, see~\cite{ausiello1984minimal}. The attribute closure $X^+$ is the set of attributes that are functionally determined by $X$. In contrast to the \fd closure, computing the closure $X^+$ of a subset of attributes $X \subseteq V$ is solvable in linear time~\cite{ramakrishnan2003database}.

A key of a relation $r$ is a subset $X \subseteq V$ that satisfies two conditions: (i) {\it uniqueness}, meaning no two distinct tuples in $r$ have identical values for the attributes in $X$, and (ii) {\it minimality}, meaning no proper subset of $X$ satisfies the uniqueness property. When $X$ is a key, the functional dependency (FD) $X \to V$ is valid, or equivalently, the closure of $X$, denoted $X^+$, is equal to $V$. A subset of attributes $X$ is considered a super-key if it satisfies the uniqueness property but not the minimality property. If a relation has more than one minimal key, each of these keys is referred to as a candidate key of $R$. Finding the least cardinality key is NP-hard. Specifically, deciding if there exists a key of cardinality $k$ is NP-complete  using a straight-forward reduction from vertex cover~\cite{lucchesi1978candidate}. Furthermore, in terms of approximation algorithm very little is known. Specifically, Akutsu and Bao~\cite{akutsu1996approximating} proved that the problem is at least as hard as the set cover problem~\cite{feige1998threshold} but they do not discuss algorithmic upper bounds.  
 
\noindent \spara{Set Cover and Red Blue Set Cover.} The      \setcover problem is a   quintessential problem in computer science, renowned for its wide applicability and fundamental role in computational complexity theory. 
It was one of Karp's 21 NP-complete problems~\cite{karp2010reducibility}, serving as a cornerstone for the study of approximation algorithms and computational intractability. The problem is defined as follows:  given a universe $U = \{u_1, u_2, \ldots, u_n\}$ and a collection of subsets $S = \{S_1, S_2, \ldots, S_m\}$ where each $S_i \subseteq U$, the      \setcover problem seeks to find a minimum subset $C \subseteq [m]$ of set indices such that $\bigcup_{i \in C} S_i = U$.    Feige showed that the      \setcover problem cannot be approximated in polynomial time to within a factor of  \( (1 - o(1)) \cdot \ln n \) unless \(\text{NP}\) has quasi-polynomial time algorithms. This inapproximability result was further improved by Dinur and Steuer who showed optimal inapproximability by proving that it cannot be approximated to \( (1 - o(1)) \cdot \ln n \) unless \( P=\text{NP} \)~\cite{
dinur2014analytical}. We use the latter result in Theorem~\ref{thm:simple}.  The      \setcover problem admits an approximation within a factor of \( O(\log n) \) utilizing either a straightforward greedy strategy or a randomized algorithm based on linear programming (LP) rounding techniques. Additionally, there is a deterministic LP-based algorithm that guarantees an \( f \)-factor approximation, with \( f \) denoting the maximum frequency an element of the universe is represented in the set collection \( S \)~\cite{williamson2011design}. The      \setcover problem is not only fundamental in computer science but also has a wide range of applications, as discussed in~\cite{cormode2010set}.

As we show, a variant of the       \setcover problem that plays an important role for the \tcand problem is the  \rbsetcover problem introduced by Carr et al.~\cite{carr1999red}: given a universe $U = R \cup B$ where $R$ and $B$ are disjoint sets representing red and blue elements respectively, and a collection of subsets $S = \{S_1, S_2, \ldots, S_m\}$ where each $S_i \subseteq U$, the  \rbsetcover problem seeks to find a  subset $C \subseteq S$ such that $\bigcup_{S_i \in C} S_i \cap B= B$ and $\bigcup_{S_i \in C} S_i \cap R$ is minimized. Recently, Chlamt{\'a}{\v{c}} et al.~\cite{chlamtavc2023approximating} proved the following state-of-the-art approximation result for the  \rbsetcover problem.

\begin{theorem}[Chlamt{\'a}{\v{c}} et al. \cite{chlamtavc2023approximating}]
\label{thm:chlamtavc1}
There exists an \(O(m^{1/3} \log^{4/3} n \log k)\)-approximation algorithm for the \rbsetcover problem where \(m\) is the number of sets, \(n\) is the number of red elements, and \(k\) is the number of blue elements.
\end{theorem}

To discuss the inapproximability of the \rbsetcover we need to introduce the Dense-vs-Random Conjecture~\cite{Bhaskara+10}.  For a graph \( G(V,E) \), the density of a subgraph induced by \( S \subseteq V \) is defined as \( \rho(S) = \frac{e(S)}{|S|} \), representing the ratio of the number of edges to the number of nodes in the subgraph~\cite{gionis2015dense}. This metric has been central to various formulations for discovering dense subgraphs~\cite{charikar2000greedy,tsourakakis2015k}. The densest $k$-subgraph problem (D$k$S) problem    asks for the densest subgraph with exactly $k$ nodes; it is  NP-hard  with the best-known approximation ratio being
$\Omega(1/n^{1/4+\epsilon})$ for any $\epsilon>0$~\cite{Bhaskara+10}. This approximability result is far off from the best-known hardness result that assumes the Exponential Time Hypothesis (ETH). If ETH holds, then D$k$S cannot be approximated within a ratio of $n^{1/(\log\log n)^c}$ for some $c>0$~\cite{manurangsi2017almost}.    Define the log-density of a graph with  $n$  nodes as $\log_n(D_{avg})$, where $D_{avg}$  represents the average degree.  The Dense-vs-Random Conjecture on graphs~\cite{chlamtavc2017minimizing} conjectures that it is hard to distinguish  between the following two cases: 1) \( G = G(n,p) \) where \( p = n^{\alpha-1} \) (and thus the graph has log-density concentrated around \( \alpha \)), and 2) \( G \) is adversarially chosen so that the densest \( k \)-subgraph has log-density \( \beta \) where \( k\beta \gg p k \) (and thus the average degree inside this subgraph is approximately \( k^{\beta} \)). In this context, $G(n,p)$ represents the random binomial graph model  ~\cite{erdHos1960evolution,frieze2023random}.

\begin{conjecture}[\cite{Bhaskara+10,chlamtac2012everywhere,chlamtavc2017minimizing}]
For all \( 0 < \alpha < 1 \), for all sufficiently small \( \varepsilon > 0 \), and for all \( k \leq \sqrt{n} \), we cannot solve \textsc{Dense vs Random} with log-density \( \alpha \) and planted log-density \( \beta \) in polynomial time (\textit{w.h.p.}) when \( \beta < \alpha - \varepsilon \).
\end{conjecture}

Using an extension of the above conjecture on hypergraphs, Chlamt{\'a}{\v{c}} et al. \cite{chlamtavc2023approximating} proved the following inapproximability result.

\begin{theorem}[Chlamt{\'a}{\v{c}} et al. \cite{chlamtavc2023approximating}]
\label{thm:chlamtavc2}
Assuming the Hypergraph Dense-vs-Random Conjecture, for every \(\varepsilon > 0\), no polynomial-time algorithm achieves better than \(O(m^{1/4-\varepsilon} \log^2 k)\) approximation for   the  \rbsetcover problem where \(m\) is the number of sets and \(k\) is the number of blue elements.
\end{theorem}

\noindent \spara{Integrality Gap.}    The integrality gap of an integer program is the worst-case ratio over all instances of the problem of value of an optimal solution to the integer programming formulation to value of an optimal solution to its linear programming relaxation~\cite{williamson2011design,vazirani2001approximation}. Notice that in the case of a minimization problem, the integrality gap satisfies 

$$ \max_{\text{instances~} Q} \frac{\textrm{OPT}_{IP}}{\textrm{OPT}_{LP}} \geq 1.$$ 

This gap provides insights into how closely the LP relaxation approximates the IP and is a measure of the performance of approximation algorithms. A small integrality gap (i.e., close to 1) indicates that the LP relaxation is a good approximation of the IP, while a large gap suggests the LP relaxation might not yield a good approximation and that other approximation strategies may be needed. It is often used as a benchmark to understand the effectiveness of approximation algorithms and to determine the best possible approximation ratio that can be achieved by any polynomial-time algorithm for NP-hard problems. Lov\'asz proved that the integrality gap for the straight-forward LP formulation of the      \setcover is $O(\log n)$~\cite{lovasz1975ratio}.

\noindent \spara{Equitable coloring.} In Section~\ref{sec:proposed}, we make use of the following theorem, known as the equitable coloring theorem, which was proved by Hajnal and Szemer\'{e}di~\cite{HajnalSzemeredi}.

\noindent 

\begin{lemma}[Hajnal-Szemer\'{e}di~\cite{HajnalSzemeredi}]
\label{lem:equi}
Every graph with $n$ vertices and maximum vertex degree at most $k$
is $k+1$ colorable with all color classes of size  $\lfloor \frac{n}{k+1} \rfloor$ or  $\lceil \frac{n}{k+1} \rceil$.
\end{lemma}

\noindent There exist efficient algorithms for finding such a coloring, e.g., ~\cite{kierstead2008short,kierstead2010fast}. The best known algorithm is due to Kierstead et al.~\cite{kierstead2010fast}  and runs in $O(kn^2)$ time and returns such an equitable coloring.  We apply the Hajnal-Szemerédi theorem to establish a Chernoff-type concentration result under conditions of limited dependencies, as demonstrated in works like \cite{kolountzakis2012efficient,pemmaraju2008randomized}.

\section{Proposed Methods}
\label{sec:proposed}

\subsection{Integer Programming Formulation}
\label{sec:ip}

 \begin{figure}
    \centering

\tikzset{every picture/.style={line width=0.75pt}} 

\begin{tikzpicture}[x=0.75pt,y=0.75pt,yscale=-1,xscale=1]

\draw  [fill={rgb, 255:red, 74; green, 144; blue, 226 }  ,fill opacity=1 ] (160,218) -- (184,218) -- (184,243) -- (160,243) -- cycle(180.4,221.6) -- (163.6,221.6) -- (163.6,239.4) -- (180.4,239.4) -- cycle ;
\draw   (116,219) -- (137,219) -- (137,242) -- (116,242) -- cycle ;
\draw   (210,220) -- (231,220) -- (231,243) -- (210,243) -- cycle ;
\draw   (339,219) -- (360,219) -- (360,242) -- (339,242) -- cycle ;
\draw   (116,176.5) .. controls (116,170.7) and (120.7,166) .. (126.5,166) .. controls (132.3,166) and (137,170.7) .. (137,176.5) .. controls (137,182.3) and (132.3,187) .. (126.5,187) .. controls (120.7,187) and (116,182.3) .. (116,176.5) -- cycle ;
\draw   (160,177.5) .. controls (160,171.7) and (164.7,167) .. (170.5,167) .. controls (176.3,167) and (181,171.7) .. (181,177.5) .. controls (181,183.3) and (176.3,188) .. (170.5,188) .. controls (164.7,188) and (160,183.3) .. (160,177.5) -- cycle ;
\draw   (209,175.5) .. controls (209,169.7) and (213.7,165) .. (219.5,165) .. controls (225.3,165) and (230,169.7) .. (230,175.5) .. controls (230,181.3) and (225.3,186) .. (219.5,186) .. controls (213.7,186) and (209,181.3) .. (209,175.5) -- cycle ;
\draw   (276,176.5) .. controls (276,170.7) and (280.7,166) .. (286.5,166) .. controls (292.3,166) and (297,170.7) .. (297,176.5) .. controls (297,182.3) and (292.3,187) .. (286.5,187) .. controls (280.7,187) and (276,182.3) .. (276,176.5) -- cycle ;
\draw   (338,177.5) .. controls (338,171.7) and (342.7,167) .. (348.5,167) .. controls (354.3,167) and (359,171.7) .. (359,177.5) .. controls (359,183.3) and (354.3,188) .. (348.5,188) .. controls (342.7,188) and (338,183.3) .. (338,177.5) -- cycle ;
\draw  [fill={rgb, 255:red, 74; green, 144; blue, 226 }  ,fill opacity=1 ] (275,221) -- (299,221) -- (299,246) -- (275,246) -- cycle(295.4,224.6) -- (278.6,224.6) -- (278.6,242.4) -- (295.4,242.4) -- cycle ;
\draw   (117,50.5) .. controls (117,44.7) and (121.7,40) .. (127.5,40) .. controls (133.3,40) and (138,44.7) .. (138,50.5) .. controls (138,56.3) and (133.3,61) .. (127.5,61) .. controls (121.7,61) and (117,56.3) .. (117,50.5) -- cycle ;
\draw   (164,51.5) .. controls (164,45.7) and (168.7,41) .. (174.5,41) .. controls (180.3,41) and (185,45.7) .. (185,51.5) .. controls (185,57.3) and (180.3,62) .. (174.5,62) .. controls (168.7,62) and (164,57.3) .. (164,51.5) -- cycle ;
\draw   (209,48.5) .. controls (209,42.7) and (213.7,38) .. (219.5,38) .. controls (225.3,38) and (230,42.7) .. (230,48.5) .. controls (230,54.3) and (225.3,59) .. (219.5,59) .. controls (213.7,59) and (209,54.3) .. (209,48.5) -- cycle ;
\draw   (279,49.5) .. controls (279,43.7) and (283.7,39) .. (289.5,39) .. controls (295.3,39) and (300,43.7) .. (300,49.5) .. controls (300,55.3) and (295.3,60) .. (289.5,60) .. controls (283.7,60) and (279,55.3) .. (279,49.5) -- cycle ;
\draw   (341,50.5) .. controls (341,44.7) and (345.7,40) .. (351.5,40) .. controls (357.3,40) and (362,44.7) .. (362,50.5) .. controls (362,56.3) and (357.3,61) .. (351.5,61) .. controls (345.7,61) and (341,56.3) .. (341,50.5) -- cycle ;
\draw   (117,126.5) .. controls (117,120.7) and (121.7,116) .. (127.5,116) .. controls (133.3,116) and (138,120.7) .. (138,126.5) .. controls (138,132.3) and (133.3,137) .. (127.5,137) .. controls (121.7,137) and (117,132.3) .. (117,126.5) -- cycle ;
\draw   (161,127.5) .. controls (161,121.7) and (165.7,117) .. (171.5,117) .. controls (177.3,117) and (182,121.7) .. (182,127.5) .. controls (182,133.3) and (177.3,138) .. (171.5,138) .. controls (165.7,138) and (161,133.3) .. (161,127.5) -- cycle ;
\draw   (210,125.5) .. controls (210,119.7) and (214.7,115) .. (220.5,115) .. controls (226.3,115) and (231,119.7) .. (231,125.5) .. controls (231,131.3) and (226.3,136) .. (220.5,136) .. controls (214.7,136) and (210,131.3) .. (210,125.5) -- cycle ;
\draw   (277,126.5) .. controls (277,120.7) and (281.7,116) .. (287.5,116) .. controls (293.3,116) and (298,120.7) .. (298,126.5) .. controls (298,132.3) and (293.3,137) .. (287.5,137) .. controls (281.7,137) and (277,132.3) .. (277,126.5) -- cycle ;
\draw   (339,127.5) .. controls (339,121.7) and (343.7,117) .. (349.5,117) .. controls (355.3,117) and (360,121.7) .. (360,127.5) .. controls (360,133.3) and (355.3,138) .. (349.5,138) .. controls (343.7,138) and (339,133.3) .. (339,127.5) -- cycle ;
\draw   (81,115.6) -- (92.5,60) -- (104,115.6) -- (98.25,115.6) -- (98.25,199) -- (86.75,199) -- (86.75,115.6) -- cycle ;

\draw (236,239) node [anchor=north west][inner sep=0.75pt]   [align=left] {{\LARGE ...}};
\draw (242,170) node [anchor=north west][inner sep=0.75pt]   [align=left] {{\LARGE ...}};
\draw (118,251.4) node [anchor=north west][inner sep=0.75pt]    {$x_{1}^{n}$};
\draw (146,252.4) node [anchor=north west][inner sep=0.75pt]    {$x_{2}^{n} =1$};
\draw (306,171) node [anchor=north west][inner sep=0.75pt]   [align=left] {{\LARGE ...}};
\draw (314,239) node [anchor=north west][inner sep=0.75pt]   [align=left] {{\LARGE ...}};
\draw (265,253.4) node [anchor=north west][inner sep=0.75pt]    {$x_{i}^{n} =1$};
\draw (340,251.4) node [anchor=north west][inner sep=0.75pt]    {$x_{n}^{n}$};
\draw (240,41) node [anchor=north west][inner sep=0.75pt]   [align=left] {{\LARGE ...}};
\draw (309,42) node [anchor=north west][inner sep=0.75pt]   [align=left] {{\LARGE ...}};
\draw (243,120) node [anchor=north west][inner sep=0.75pt]   [align=left] {{\LARGE ...}};
\draw (307,121) node [anchor=north west][inner sep=0.75pt]   [align=left] {{\LARGE ...}};
\draw (121,14.4) node [anchor=north west][inner sep=0.75pt]    {$x_{1}^{0}$};
\draw (165,16.4) node [anchor=north west][inner sep=0.75pt]    {$x_{2}^{0}$};
\draw (342,13.4) node [anchor=north west][inner sep=0.75pt]    {$x_{n}^{0}$};
\draw (283,15.4) node [anchor=north west][inner sep=0.75pt]    {$x_{i}^{0}$};
\draw (118,189.9) node [anchor=north west][inner sep=0.75pt]    {$x_{1}^{n-1}$};
\draw (159,188.9) node [anchor=north west][inner sep=0.75pt]    {$x_{2}^{n-1}$};
\draw (337,188.9) node [anchor=north west][inner sep=0.75pt]    {$x_{n}^{n-1}$};
\draw (273,193.9) node [anchor=north west][inner sep=0.75pt]    {$x_{i}^{n-1}$};
\draw (119,90.4) node [anchor=north west][inner sep=0.75pt]    {$x_{1}^{n-D}$};
\draw (162,88.9) node [anchor=north west][inner sep=0.75pt]    {$x_{2}^{n-D}$};
\draw (279,91.9) node [anchor=north west][inner sep=0.75pt]    {$x_{i}^{n-D}$};
\draw (341,93.9) node [anchor=north west][inner sep=0.75pt]    {$x_{n}^{n-D}$};
\draw (212,76.4) node [anchor=north west][inner sep=0.75pt]    {$\vdots $};
\draw (121,71.4) node [anchor=north west][inner sep=0.75pt]    {$\vdots $};
\draw (118,142.4) node [anchor=north west][inner sep=0.75pt]    {$\vdots $};
\draw (341,73.4) node [anchor=north west][inner sep=0.75pt]    {$\vdots $};
\draw (340,143.4) node [anchor=north west][inner sep=0.75pt]    {$\vdots $};

\end{tikzpicture}
    \caption{\label{fig:visual}  Visual representation of IP~\eqref{ipexact}. }
\end{figure}
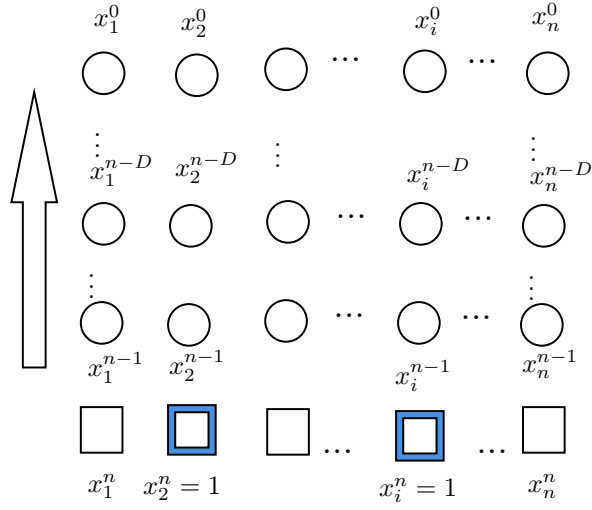


We formulate the \tcand problem as an integer program  (IP)  that provides an optimal solution. This formulation serves as the basis of our LP-based approximation algorithms. Notationally, we define $[k,n] \myeq \{k,k+1,\ldots,n\}$ and $[n]=[1,n]$. We assume, without loss of generality, that each \fd's right-hand side has a size exactly equal to $1$. Figure~\ref{fig:visual} visualizes the variables introduced by our IP. Specifically, we introduce a set $n^2+n$ variables  $x^{d}_i$ for $d \in [0,n], i\in [n]$.  The bottom level  corresponds to the set of Boolean variables $x_i^n$. We constraint the set of variables in this layer corresponding to the target variables (blue dotted circles) to be equal to 1.  In general, the set of variables $\{x_i^{n-D}\}$ will indicate what variables we should include in our output if we allow for $D$ rounds of inference, for $D \in [n]$. For any layer \( k = n, \ldots, 1 \), a variable \( x_i^k \) will be set to 1 if and only if either \( x_i^{k-1} = 1 \) or if there exists an \fd of the form \( i_1, \ldots, i_r \to i \) where all the variables on the left side in the previous layer \( k-1 \) are equal to 1, i.e., \( x_{i_1}^{k-1} = \ldots = x_{i_r}^{k-1} = 1 \). In this case, we shall say that the \fd  is activated.  We can express this logic   as 

$$
x_i^k = \textbf{OR}\left(x_i^{k-1}, \textbf{AND}\left(x_{i_1}^{k-1},\ldots,x_{i_r}^{k-1}\right), \ldots\right),
$$

\noindent where we have a \textbf{AND} term for each \fd of the form \( i_1\ldots i_r \to i \) as linear constraints. The logical operators \textbf{OR}, \textbf{AND} can be easily expressed as linear constraints~\cite{wolsey2020integer}. Towards this formulation, we introduce a new set of variables $z^{(d)}_{LS}$ which corresponds to whether \fds with left hand side $LS$ can be activated in round $d$,  i.e., all the variables participating in $LS$ from the previous layer are already set to 1. Putting those constraints together along with the objective function, we arrive at the following IP for $D$-round \tcand where $1 \leq D \leq n$.

\framebox{
  \begin{minipage}{0.96\linewidth}
\begin{equation}
\label{ipexact}
\begin{array}{ll@{}ll}
\text{minimize}  & \displaystyle\sum\limits_{j=1}^{n} x_j^{n-D}   &\\
\text{subject to}& \displaystyle  
x_i^{d} \leq x_i^{d-1} + \sum_{LS:LS \to i} 
z_{LS}^{d} \text{~~~} \forall d \in [n-D+1,n], i \in [n]\\ 
&z_{LS}^{d} \leq x_j^{d-1} \text{~~~}  \forall LS \to i, \forall j \in LS, d \in [n-D+1,n]
\\
             &              x_i^{n}=1 \text{~~~} \forall i \in T                      \subseteq [n]     & \\
\end{array}
\end{equation}
\end{minipage}
}

\vspace{3mm}
 By setting $D=n$, the IP is guaranteed to return an optimal solution.  This is because  applying the \fd rules a maximum of $n$ times is sufficient; if an iteration fails to fix any additional variables, subsequent iterations will not alter the outcome and   a variable, once set, does not change.

\subsection{Simple \fds} 
\label{sec:simplefds} 

Before addressing the general case of the problem, we focus on an important special case in this section. Notice that without any loss of generality, the right side of any $\fd$ consists of a single attribute.  
We denote the set of all left-sides of \fds as $\mathcal{LS} = \{ LS \subseteq V: (LS \to i) \in \mathcal{F}, i \in V \}$. If an \fd has a single variable on the left side, we refer to this \fd as {\it simple}. Otherwise, we refer to it as {\it non-simple}.  We call a set $\mathcal{F}$ of \fds simple if all \fds in $\mathcal{F}$ are simple.  If there exists at least one non-simple \fd in $\mathcal{F}$, we refer to $\mathcal{F}$ as {\it regular} or {\it non-simple}. For example, the set of \fds $\{ a\to b, b\to c, c \to d \}$ is simple because all \fds are simple. The set of \fds $\{ a\to b, bc\to d \}$ is regular/non-simple since the \fd $bc\to d$ is non-simple. We also define  the following natural graph for any set of \fds.  It is worth mentioning that similar notions exist in the literature, see~\cite{ausiello1983graph} and \cite[Definition 3]{saiedian1996efficient}. 

\begin{definition}
Let $\mathcal{F}$ be a set of \fds. We define the corresponding \fd-graph $G_\mathcal{F}$ (or simply $G$) as follows: for each variable $i$ appearing in $\mathcal{F}$ we create a vertex $v_i$. There exists a directed edge  $(v_i, v_j)$ when there exists an \fd of the form $X\to y$ with $i \in X, y=j$. 
\end{definition}

Designing approximation algorithms is easier for simple \fds. Intuitively, when dealing with a simple set of \fds, the \fd-graph precisely mirrors the input data.  For example, if the \fds are simple and the \fd-graph is strongly connected, any single variable $i$ is an optimal solution for any target set $T$,  as $\{i\}^+=V$. However, in general the \fd-graph may not be strongly connected. However, it is a well-known fact that any directed graph can be decomposed as a directed acyclic graph (DAG) of \emph{strongly connected components} (SCCs)~\cite{tarjan1972depth} and this decomposition requires linear time in the size of the graph. We use this fact to prove a refined approximation result and establish that that the lower bound established by Akutsu and Bao~\cite{akutsu1996approximating} is tight, see also~\cite{feige1998threshold}. We present this result in the subsequent theorem.


\begin{theorem}
\label{thm:simple}
Let $\mathcal{F}$ be a set of simple \fds and let $g$ be the number of the strongly connected components (SCCs) of $G$. Then there exists a polynomial time $\ln g$-approximation algorithm for the \tcand problem. Furthermore, it cannot be approximated within $(1-o(1)) \ln g$ unless $P=NP$.
\end{theorem}

\begin{proof}
We define for each node $v_i$ in the \fd-graph $G$, the set $\mathrm{Reach}(v_i)$ of nodes that are reachable from $v_i$. Then, $S_i := \mathrm{Reach}(v_i) \cap T$ is the set of target variables reachable from $v_i$.  Notice that for any $v_i, v_j$ that belong to the same SCC the sets $S_i, S_j$ are the same, i.e.,  $S_i=S_j$. Thus, we define $\hat{S}_i, i=1,\ldots,g$ to be the set of target variables reachable from each SCC.  It is straight-forward to observe that solving the \setcover problem for the instance defined by $\{\hat{S}_1 \cap T,\ldots,\hat{S}_g \cap T\}$ and universe $T$ yields the optimal solution. This allows us to use the well known $\ln g - \ln \ln g + \Theta(1)$-approximation greedy algorithm for  the \setcover~\cite{slavik1996tight}.  By Dinur and Steuer, the \tcand for simple \fds cannot be approximated to \( (1 - o(1)) \cdot \ln g \) unless \( P=\text{NP} \).
\end{proof}

\noindent   Notice that the number of SCCs $g$ can be significantly less than $n$.   Furthermore, our proof directly yields that  the above bound can be further tightened to the logarithm of the number of sources in the DAG of the \fd-graph, which is trivially upper bounded  by the number of SCCs. This is the case for it suffices to pick at most one node from each strongly connected component of $G$ that is a source; this can be any node.  We state this as the following corollary.

\begin{corollary}
Let $\mathcal{F}$ be a set of simple \fds and let $s$ be the number of the source nodes in the SCC DAG of the fd-graph $G$. Then there exists a polynomial time $\ln s$-approximation algorithm for the \tcand problem. Furthermore, it cannot be approximated within $(1-o(1)) \ln s$ unless $P=NP$.
\end{corollary} 


\subsection{LP Relaxation and Approximation Algorithms}
\label{sec:approx} 

Our building block for the $D$-round \tcand problem is solving the 1-round \tcand problem and then iterating this algorithm for $D$ layers. Given its importance and for the reader's convenience, we introduce it as a special case with simplified notation compared to the IP~\eqref{ipexact}. 

\noindent \spara{Single layer/1-round  \tcand.} Our goal in the single-layer problem is to choose a set of variables $S \subseteq [n]$ such that their {\em one step} closure includes the target variables. By the term {\em one step} closure, we mean that an attribute $i$ is active if it is either included in $S$ or if there exist attributes $i_1,\ldots,i_r$ such that they are all active (i.e., in $S$) and there exists an \fd in $\mathcal{F}$ of the form $i_1i_2\ldots i_r \to i$.   
The bottom row with the squares encodes the target attributes; for each attribute we have a Boolean variable $x_i$ and this is set to $1$ for each target variable (blue filled squares); the rest may be 0 or 1 depending on which set of variables we will activate on the top row. To encode the   output set $S$ we define Boolean variables $y_i$ for $i=1,\ldots,n$. Our goal is to minimize the number of variables we include in $S$ or in terms of the $y$ variables the sum $\sum_{i=1}^n y_i$. The constraint of covering the target variables is expressed as $x_i=1$ for each target variable $i \in T$. The connection between the $y$ and $x$ variables --as explained also earleir-- is expressed as follows  
$x_i =  \textbf{OR}(y_i,  \textbf{AND}(y_{i_1},\ldots,y_{i_r}), \ldots)$ 
where we include in the {\bf OR} all \fds of the form $i_1\ldots i_r \to i$ as the {\bf AND} of the corresponding left-side variables $i_1,\ldots,i_r $.  Figure~\ref{fig:onelayer} illustrates a single-layer version of the \tcand problem. Each column corresponds to an attribute $i$, $i=1,\ldots,n$.

\begin{figure}
\centering
\tikzset{every picture/.style={line width=0.75pt}} 

\begin{tikzpicture}[x=0.75pt,y=0.75pt,yscale=-1,xscale=1]

\draw  [fill={rgb, 255:red, 74; green, 144; blue, 226 }  ,fill opacity=1 ] (141,164) -- (165,164) -- (165,189) -- (141,189) -- cycle(161.4,167.6) -- (144.6,167.6) -- (144.6,185.4) -- (161.4,185.4) -- cycle ;
\draw   (97,165) -- (118,165) -- (118,188) -- (97,188) -- cycle ;
\draw   (191,166) -- (212,166) -- (212,189) -- (191,189) -- cycle ;
\draw   (320,165) -- (341,165) -- (341,188) -- (320,188) -- cycle ;
\draw   (97,133.5) .. controls (97,127.7) and (101.7,123) .. (107.5,123) .. controls (113.3,123) and (118,127.7) .. (118,133.5) .. controls (118,139.3) and (113.3,144) .. (107.5,144) .. controls (101.7,144) and (97,139.3) .. (97,133.5) -- cycle ;
\draw   (141,134.5) .. controls (141,128.7) and (145.7,124) .. (151.5,124) .. controls (157.3,124) and (162,128.7) .. (162,134.5) .. controls (162,140.3) and (157.3,145) .. (151.5,145) .. controls (145.7,145) and (141,140.3) .. (141,134.5) -- cycle ;
\draw   (190,132.5) .. controls (190,126.7) and (194.7,122) .. (200.5,122) .. controls (206.3,122) and (211,126.7) .. (211,132.5) .. controls (211,138.3) and (206.3,143) .. (200.5,143) .. controls (194.7,143) and (190,138.3) .. (190,132.5) -- cycle ;
\draw   (257,133.5) .. controls (257,127.7) and (261.7,123) .. (267.5,123) .. controls (273.3,123) and (278,127.7) .. (278,133.5) .. controls (278,139.3) and (273.3,144) .. (267.5,144) .. controls (261.7,144) and (257,139.3) .. (257,133.5) -- cycle ;
\draw   (319,134.5) .. controls (319,128.7) and (323.7,124) .. (329.5,124) .. controls (335.3,124) and (340,128.7) .. (340,134.5) .. controls (340,140.3) and (335.3,145) .. (329.5,145) .. controls (323.7,145) and (319,140.3) .. (319,134.5) -- cycle ;
\draw  [fill={rgb, 255:red, 74; green, 144; blue, 226 }  ,fill opacity=1 ] (256,167) -- (280,167) -- (280,192) -- (256,192) -- cycle(276.4,170.6) -- (259.6,170.6) -- (259.6,188.4) -- (276.4,188.4) -- cycle ;

\draw (220,180) node [anchor=north west][inner sep=0.75pt]   [align=left] {{\LARGE ...}};
\draw (223,127) node [anchor=north west][inner sep=0.75pt]   [align=left] {{\LARGE ...}};
\draw (99,96.4) node [anchor=north west][inner sep=0.75pt]    {$y_{1}$};
\draw (143,96.4) node [anchor=north west][inner sep=0.75pt]    {$y_{2}$};
\draw (322,99.4) node [anchor=north west][inner sep=0.75pt]    {$y_{n}$};
\draw (99,197.4) node [anchor=north west][inner sep=0.75pt]    {$x_{1}$};
\draw (127,198.4) node [anchor=north west][inner sep=0.75pt]    {$x_{2} =1$};
\draw (193,198.4) node [anchor=north west][inner sep=0.75pt]    {$x_{3}$};
\draw (192,96.4) node [anchor=north west][inner sep=0.75pt]    {$y_{3}$};
\draw (287,128) node [anchor=north west][inner sep=0.75pt]   [align=left] {{\LARGE ...}};
\draw (291,182) node [anchor=north west][inner sep=0.75pt]   [align=left] {{\LARGE ...}};
\draw (259,96.4) node [anchor=north west][inner sep=0.75pt]    {$y_{i}$};
\draw (246,199.4) node [anchor=north west][inner sep=0.75pt]    {$x_{i} =1$};
\draw (321,197.4) node [anchor=north west][inner sep=0.75pt]    {$x_{n}$};

\end{tikzpicture}
\caption{\label{fig:onelayer} The Boolean variables $x_1,\ldots,x_n$ denote whether an attribute $i$ is {\it active}, i.e.,  meaning it is included in the closure of the selected set of variables.  The target variables are all set to 1.  }
\end{figure}
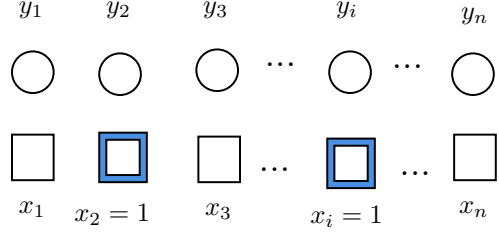

\noindent \spara{LP Relaxation.}
We explicitly express the linear programming relaxation of the integer program $\eqref{ipexact}$, incorporating a single layer or round of  \fd inference. In Sections~\ref{subsec:deterministic} and \ref{subsec:randomized} 
we present a deterministic and randomized based on the following LP relaxation and we show how we apply it for the $D$-round \tcand problem. For simplicity, we assume, without any loss of generality, that  the set of \fds includes the valid \fds $i \to i$ for each $i \in [n]$.

\framebox{%
  \begin{minipage}{0.96\linewidth}
\begin{equation}
\label{lp:oneround}
\begin{array}{ll@{}ll}
\text{minimize}  & \displaystyle\sum\limits_{j=1}^{n} & y_{j} &\\
\text{subject to}& \displaystyle y_i + \sum\limits_{LS \to i}   & z_{LS} \geq 1,  & \forall i \in T\\
                 & \displaystyle     & z_{LS} \leq y_j,  & \forall j \in LS  \text{~~where~~}  LS \to i,   i \in T \\
                                  
                                  
                 &                                                &y_{j}, z_{LS} \in [0,1], & \forall j \in V, \forall LS \text{~~~of the form~~} LS \to i,  i \in T 
\end{array}
\end{equation}
\end{minipage}}

\vspace{2mm}

\subsubsection{Deterministic rounding} 
\label{subsec:deterministic} 

Our deterministic rounding approximation algorithm for the single-round  \tcand problem solves the LP relaxation $\eqref{lp:oneround}$ and outputs all attributes $i$ for which the corresponding variable $y_i$ is at least a certain threshold. The threshold value is determined by the input.  Define $f_i$ as the number of input \fds of the form $X \to i$, represented by $f_i = |{ X \subseteq V: X \to i \in \mathcal{F} }|$. Let $f$ denote the maximum value of $f_i$ across all elements in $V$, i.e., $f = \max_{i \in V} f_i$.   In simple terms, \( f \) represents the maximum number of variables on the left side of any \fd in our collection of \fds, $\mathcal{F}$.
The threshold value is set equal to
 $\frac{1}{f+1}$. This simple process is outlined for completeness as Algorithm~\ref{alg:lprounding} and the approximation guarantee states as Theorem~\ref{thm:fplus1}.

  \begin{algorithm}
\begin{algorithmic}
\STATE Solve the LP relaxation~\eqref{lp:oneround}
\STATE $I \leftarrow \{ i : y_i   \geq \frac{1}{f+1} \}$
\STATE Output $I$
\end{algorithmic}
\caption{ \label{alg:lprounding} $(f+1)$ approximation algorithm for the 1-round \tcand (Problem~\ref{prob:rtcand} with $D=1$)}
\end{algorithm}

\begin{theorem}
\label{thm:fplus1}
Algorithm~\ref{alg:lprounding} is an $(f+1)$ approximation for the 1-round \tcand problem.     
\end{theorem}

\begin{proof}

\underline{Feasibility:} First we prove that $T \subseteq I^+$. We note that for all $i \in T$ the inequality 
\[  y_i + \sum\limits_{LS \to i}     z_{LS} \geq 1 \]

\noindent   implies that at least one of the $f+1$ summands will be at least $\frac{1}{f+1}$. If that summand is $y_i$, then we add $i$ to $S$, and thus $i$ is trivially in the closure of $S^+$.   Otherwise, $z_{LS} \geq \frac{1}{f+1}$ for some $LS$ and due to the linear constraints $y_j \geq \frac{1}{f+1}, \forall j \in LS$. This means that all such $j$ will be added to $I$ and therefore $i$ will be in the closure $I^+$.

\noindent \underline{Approximation guarantees:} The cost of the solution is upper bounded as follows:

\[ |I| \leq (f+1) \cdot \sum_{i \in I} y_i = (f+1) \cdot OPT_{LP} \leq (f+1) \cdot OPT_{IP}.\]
 
This  completes our proof.
\end{proof}

\noindent We now state our main result for the $D$-round \tcand problem as  Theorem~\ref{thm:lpapprox}. 
Our proof is constructive. To solve the $D$-round \tcand problem, we iteratively apply Algorithm~\ref{alg:lprounding} $D$ times, allowing for $D$ rounds of \fd inference. Our algorithm is described as Algorithm~\ref{alg:lproundingD}. We return to the indexing notation previously used in Figure~\ref{fig:visual}. Observe that in this case the objective becomes the minimization of $\sum\limits_{j=1}^{n} x_j^{n-D}$ where $ D\leq n$.  Next, we present Theorem~\ref{thm:lpapprox}, which demonstrates how the approximation algorithm we developed for the 1-round \tcand can be applied to the D-round \tcand.

\begin{algorithm}
\begin{algorithmic}
\STATE Solve the LP relaxation~\eqref{ipexact} to obtain fractional values for all variables $x,z$
\STATE $ I \leftarrow \{ i : x_i^{n-D}  \geq \frac{1}{(f+1)^D} \}$
\STATE Output $I$
\end{algorithmic}
\caption{ \label{alg:lproundingD} $(f+1)^D$ approximation algorithm for the $D$-round \tcand (Problem~\ref{prob:rtcand})}
\end{algorithm}

\begin{theorem}(Approximating \tcand)
\label{thm:lpapprox}
There exists a polynomial time $(f+1)^D$-approximation algorithm for the $D$-round \tcand problem.
\end{theorem}

 \begin{proof}[Proof]
Similar to the proof of Theorem~\ref{thm:fplus1}, we first establish feasibility, ensuring that the closure of the output set $I$ contains all target variables $T$, and then we prove the approximation guarantee.

\noindent \underline{Feasibility:} 
For each target variable $i \in T$, $x_i^n=1$. By the LP inequality, 

\begin{equation}
\label{eq:ineq}
      1 = x_i^{n} \leq x_i^{n-1} + \sum_{LS:LS\to i} z_{LS}^{(n-1)} 
\end{equation}

\noindent with the same reasoning as in Theorem~\ref{thm:fplus1},    either 
 $x_i^{n-1} \geq \frac{1}{f+1}$ or $z_{LS}^{n-1} \geq \frac{1}{f+1}$
 for some $LS\to i$.   The latter inequality yields that 
 \[ x_j^{n-1} \geq \frac{1}{f+1}, \forall j \in LS.\]
 
 Thus, either of the sets $  \{x_i^{n-1}\}, \{ x_j^{n-1} \}_{j \in LS}$  for some \fd $LS \to i$ must have all its elements larger than $\frac{1}{f+1}$.   We apply the same reasoning for the previous layer $n-2$, with the difference that the left-hand-side of Inequality~\ref{eq:ineq} is $\frac{1}{f+1}$. This implies that   for a variable 
 $x_j^{n-1}$ that is at least $\frac{1}{f+1}$ either $x_j^{n-2} \geq \frac{1}{(f+1)}\times \frac{1}{(f+1)}=\frac{1}{(f+1)^2}$ or 
 $x_{j'}^{n-2} \geq \frac{1}{(f+1)^2}$ for all $j' \in LS'$ for some $LS' \to j$ with $z_{LS'} \geq \frac{1}{(f+1)^2}$. 
Using  backwards induction and the same averaging argument, we obtain that $I$ is a feasible solution.


\noindent \underline{Approximation guarantees:} We simply observe that the output size $|I|$ satisfies

 \[ |I| \leq (f+1)^D \cdot \sum_{i \in I} x_i \leq (f+1)^D OPT_{LP} \leq (f+1)^D OPT_{IP},\]

\noindent which yields the desired bound.
 
 \end{proof}

\subsubsection{Randomized rounding} 
\label{subsec:randomized} 

Our randomized algorithm relies again on solving the LP relaxation~\eqref{lp:oneround} to obtain $\{ y_i\}_{i \in [n]}$ values and   on the KKMS algorithm~\cite{kierstead2010fast} as a subroutine for finding an equitable Hajnal-Szemer\'{e}di partition as described in Lemma~\ref{lem:equi}.  The algorithm is outlined in pseudocode as Algorithm~\ref{alg:randomizedrounding}. Let $\mathcal{LS}$ be the set of all  left-side sets of variables that appear in $\mathcal{F}$, i.e., $\mathcal{LS}=\{LS: LS \to i \in \mathcal{F} \}$. Define $\Delta = \max_{LS \in \mathcal{LS}} |\{LS' \in \mathcal{LS}: LS \cap LS' \neq \emptyset\}|$ denote the maximum number of \fds that share at least one common attribute with any \fd in $\mathcal{F}$.  The algorithm initiates a set $OUT$ that will contain a set of  variables whose closure will contain the target set $T$ with high probability. The algorithm considers each target element separately.  To determine which variables we will include in the set $OUT$ we use   the constructive polynomial time algorithm for Hajnal-Szemer\'{e}di~\cite{HajnalSzemeredi} lemma~\ref{lem:equi}. This allows us to find partition the \fds into sets  whose left sides share no attribute. This  ensures that the joint distribution of those \fds factor over  the individual left sides due to independence; we elaborate more on this in the following paragraph.  Among those sets we choose one set with the property that the sum of $z_{LS}$ values is at least $\frac{1}{\Delta+1}$; such a set is guaranteed to exist by a simple averaging argument.

Specifically, let us consider the solution to the 1-round \tcand LP relaxation, yielding fractional values $y_1, \ldots, y_n$. Focusing on a specific target element $t \in T$, we define $\mathcal{F}_t = \{LS \in \mathcal{LS}: LS \to t \}$ as the set of \fds with variable $t$ on their right-hand side. Viewing the left sides of these \fds as a collection of hyperedges, our objective is to randomly select attributes such that at least one hyperedge   ``survives'' after sampling,  namely all the variables are included in $OUT$. This ensures
 that $t$ is included in the closure of the randomized output.  
 As mentioned earlier, analyzing this randomized procedure involves dependencies; the joint distribution of the survival  of two hyperedges does not factor over the variables in them, since they may overlap.  
 e.g.,  $X \to t$ and $Y \to t$, share attributes, i.e., $X \cap Y \neq \emptyset$. This interdependence adds complexity to the analysis of a randomized rounding approach akin to the \setcover algorithm~\cite[Section 1.7]{williamson2011design} and we address it using Lemma~\ref{lem:equi}. We state our main result as the next theorem.

 \begin{algorithm}
\begin{algorithmic}
\STATE Solve the LP relaxation~\eqref{lp:oneround} to obtain $\{ y_i\}_{i \in [n]}$ values
\STATE Compute $\Delta = \max_{LS \in \mathcal{LS}} |\{LS' \in \mathcal{LS}: LS \cap LS' \neq \emptyset\}|$
\STATE $OUT \leftarrow \emptyset$
\FOR{ each target element $t \in T$}
\STATE Find a set $S_{j^\star}$ of \fds $\{ LS \to t\}$  with the properties that (i) no two \fds share an attribute and (ii) whose sum of $z_{LS} \geq \frac{1}{\Delta+1}$ using the KKMS algorithm~\cite{kierstead2010fast}.
\STATE For each variable $k \in S_{j^\star}$, toss $2  (\Delta+1) \log n$ coins each with success probability $y_k$. 
\IF{ success at least once for variable $k$ } 
\STATE $OUT \leftarrow OUT \cup \{k\}$ 
\ENDIF 
\ENDFOR
\STATE Return $OUT$
\end{algorithmic}
\caption{ \label{alg:randomizedrounding} $2\log n (\Delta+1)$-approximation algorithm for the $1$-round-\tcand (Problem~\ref{prob:rtcand} with $D=1$)}
\end{algorithm}

\begin{theorem}
\label{thrm:randomizedrounding} 
Then, there exists a polynomial-time $c(\Delta+1)\log n$-approximation algorithm that solves the 1-round \tcand problem with high probability, where $c$ is a constant.
\end{theorem}
 
\begin{proof} 
Define $\mathcal{B}_i$ to be the bad event that target element $i$ is not covered by Algorithm~\ref{alg:randomizedrounding}. Fix any target element $i$ and consider a meta-graph $G$ where each node represents the left-side of some \fd $LS \to i$ and two nodes $LS_j, LS_k$ are connected iff $LS_j \cap LS_k \neq \emptyset$. Recall,  $\Delta = \max_{LS} |\{ LS' : LS \cap LS' \neq \emptyset \} |$ and thus the maximum degree in $G$ is upper-bounded by $\Delta$. We invoke the equitable coloring theorem~\ref{lem:equi} on $G$ to obtain color classes $S_1,\ldots,S_{\Delta+1}$ of size (essentially) $ \frac{n}{\Delta+1}$. By grouping the terms $z_{LS}$ according to color classes we obtain 
$\sum_{j=1}^{\Delta+1} \sum\limits_{LS \in S_j} z_{LS} \geq 1.$  For at least one of the color classes $j$ the summation term $ \sum\limits_{LS \in S_j} z_{LS} \geq \frac{1}{\Delta+1}$. Let $j^\star$ be such an index. Observe that all the \fds within the $S_{j^\star}$ share no attributes since by the equitable coloring theorem they form an independent set in the meta-graph $G$. Thus, we obtain from the independence of the coin tossing 

\begin{align*}
 \Prob{i\text{~not activated}} &= \prod_{LS \in S_{j^\star} } \Bigg( 1-\prod_{k \in LS} y_k \Bigg) \leq \prod_{LS \in S_{j^\star}} e^{-\prod_{k \in LS} y_k} =\\ 
 &= \exp\Big(-\sum_{LS \in S_{j^\star}} \prod_{k \in LS} y_k  \Big)   \leq e^{-\frac{1}{\Delta+1}}.
\end{align*}

For each attribute $j$ we toss a biased coin with probability $y_j$ of success  $c (\Delta+1) \log n$ times independently where $c>1$ is a constant. If success occurs at least once, we include $j$ in our output.  The probability $p_j$ that an element $j$ is included in the output  satisfies $p_j=1-(1-y_j)^{c (\Delta+1) \log n} \leq c (\Delta+1)  \log n y_j$. It is straight-forward to show that using this procedure 

$$ \Prob{\mathcal{B}_i} \leq e^{-\frac{c \log n (\Delta+1)}{\Delta+1}} = \frac{1}{n^c}.$$

Furthermore, the expected cost of Algorithm's~\ref{alg:randomizedrounding} is upper bounded as follows: 

\begin{align*}
\Mean{\text{OUTPUT COST}} &\leq \sum_{i=1}^n p_i \leq  \sum_{i=1}^n(c  (\Delta+1) \log n) y_i= c  (\Delta+1) \cdot \log n \cdot OPT_{LP} \leq \\ 
&\leq (c  (\Delta+1) \log n) \cdot OPT_{IP}.
\end{align*}

Since $\Prob{\exists i\in T \text{~not activated~}} = \Prob{\cup_i \mathcal{B}_i}$, by a union bound we conclude that all target variables are activated with high probability for any constant $c>1$: 

\begin{align*}
\Prob{\cup_i \mathcal{B}_i} &\leq |T| 
\max_{i \in T} \Prob{i \text{~not activated}} \leq n \frac{1}{n^c} = o(1).
\end{align*}

Using the rule of conditional probability  and the fact that the  good event $ \cap \bar{\mathcal{B}_i}$ (i.e., all target variables are covered) holds {\it whp},  we obtain the desired result 

$$\Mean{\text{OUTPUT COST}|\bar{\mathcal{B}} } \leq O((\Delta+1)\log n) OPT_{IP}.$$
\end{proof}

We apply our randomized procedure for $D$ layers in the same manner as the deterministic algorithm, achieving an approximation guarantee of $\big(c(\Delta+1)\log n\big)^D$. The only distinction from the deterministic approximation algorithm is ensuring a success probability of $1-o(1)$. This requirement is easily met, as the failure probability for a single application of the \fd rules is $\frac{1}{n^c}$ for some sufficiently large constant $c$. By applying a union bound over $D$ rounds, since $D\leq n$, we achieve the desired result.

\spara{Remarks.} Solving the \tcand problem using an integer program (IP) can be a practical approach for small to medium-scale instances. In a query optimizer, where speed is crucial, this method can be effectively applied to smaller instances.  Between the two approximation algorithms we presented, the deterministic algorithm is more straight-forward to implement, as the randomized algorithm relies on the complex algorithm due to Kierstad et al.\cite{kierstead2010fast} for finding an equitable coloring on a meta-graph of the left-sides of \fds.  From an approximation guarantee perspective, the two approximation algorithms cannot be directly compared based due to different parameterizations. While the parameters clearly cannot take arbitrary values  (e.g., the number of \fds $|\mathcal{F}|$ is upper bounded by $n \cdot f$ given that each variable $i \in [n]$ participates in at most $f$ \fds and $\Delta \leq |\mathcal{F}|-1$), the values $f+1, (\Delta+1)\log n$ can be either larger or smaller depending on the specific instance.We present two extreme scenarios to illustrate this claim, though such situations are unlikely to occur in practice.  For example,  when $f=\Theta(|\mathcal{F}|)=O(n)$ (e.g., a constant number of variables are on the RS of $f=O(n)$ \fds per each), and $\Delta=O(1)$ (e.g. the left-sides are singleton sets), the value $(\Delta+1) \log n $ is less than  $f+1$ asymptotically. On other other extreme, there exist instances where $\Delta = \Theta(|F|)$ (e.g., there exists one common variable to all the left-sides of \fds) and $f$ is a low value making the value $f+1$ smaller than  $(\Delta+1) \log n $. In relatively large practical instances, these extreme scenarios do not occur. However, the parameters generally lean towards the deterministic algorithm side, as \(\Delta\) tends to be large due to the presence of a few common variables on the left side of most of the \fds, while $f\ll |\mathcal{F}|$. For these reasons, we find the deterministic algorithm to be more practical. However, since it requires solving an LP, there is still an open area for developing well-performing heuristics for systems.



\subsubsection{$1$-round \tcand is equivalent to\rbsetcover}
\label{sec:redblue}

We discover that the general case (regular \fds) for 1-round \tcand problem the problem is equivalent to a variant of 
the \setcover problem, referred to as \rbsetcover~\cite{goldwasser1997intractability, alekhnovich2001minimum,carr1999red}. Our discovery is important for two reasons: (i) from the equivalence reduction we obtain a new algorithm from \cite{chlamtavc2023approximating} and (ii) an inapproximability result, showing that polynomial time is very likely to restrict the approximation factor of the problem to $|\mathcal F|^{\frac{1}{4}}$ at best. It is important to note that the number of input \fds, denoted as \( |\mathcal{F}| \), can range from a constant to an exponential function of \( n \). Therefore, our bound does not directly provide a limit based solely on \( n \). Instead, it offers a new form of approximation guarantee. Our findings are summarized in the following theorem.

\begin{theorem} 
\label{thm:oneround}(Algorithm and Inapproximability of $1$-round \tcand)
The $1$-round \tcand problem admits a polynomial time $\widetilde{O}(|\mathcal F|^{\frac13})$-approximation algorithm and, additionally, does not admit a polynomial time algorithm with approximation ratio better than $\widetilde{O}(|\mathcal F|^{\frac14 - \epsilon})$ unless the Dense-vs-Random conjecture~\cite{chlamtavc2017minimizing} fails. The hardness results carries over to $D$-round \tcand for any $D \in [n]$.
\end{theorem}

We now prove Theorem~\ref{thm:oneround} by proving an equivalence between the $1$-round \tcand problem and\rbsetcover and then utilizing the recent progress in~\cite{chlamtavc2023approximating}. We adopt the same convention regarding the use of red/blue colors as described in the \rbsetcover problem in Section~\ref{sec:related}.

\begin{proof}
Given a collection of \fds $\mathcal{F}$ and a set $T$ we create an instance of \rbsetcover as follows. Firstly, for each element $i \in T$, we introduce a new variable $i^{(\text{new})}$. Our universe $U$ will be composed of these new variables in addition to the original ones.
The new variables will be colored blue whereas the old variables will be colored red. Now, for every \fd $LS \to i$ we create a set $S_{LS} := LS \cup \{i^{(new) }\}$ (note that due to the $1$-round scenario, we may discard each \fd of the form $LS \to i$ with $i \notin T$ so $i^{(new)}$ is well defined). We claim that any solution to \rbsetcover of size $k$ yields a solution to $1$-round \tcand of size $k$ and vice versa. Indeed, if we have a collection of sets $\mathcal{S}$ that cover all blue points and $k$ red points, we may as well pick the variables corresponding to those red points as our solution. By definition, for every $S_{LS} \in \mathcal{S}$ we have covered LS and so this means that all blue points are covered as there exists an \fd of the form $LS \to i$ $\forall i \in T$ and some totally covered LS, which correspond exactly to $T$ being inferred. Additionally, every solution to $1$-round \tcand of size $k$ can be mapped to a \rbsetcover solution by noting that every chosen variable corresponds to a red point and thus we may pick all the sets for which all their red variables are chosen in the solution. This covers all the blue points by the fact that all variables in $T$ can be inferred from the original instance and every blue point corresponds to a variable in $T$.
For the converse, consider an instance of \rbsetcover where $U \myeq R \cup B$ is the union of a red set $R$ and a blue set $B$. We can construct an instance of $1$-round \tcand as follows: For each set $S \myeq \{e_1,\ldots, e_t\}$, we create a functional dependency (\fd) $S \setminus B \rightarrow S \cap B$. We then set $T \myeq B$ and solve the $1$-round \tcand problem for this instance. Note that this process may generate \fds of the form $\emptyset \rightarrow V'$ for some non-empty set $V'$.

A solution to the constructed $1$-round \tcand instance of size $k$ means that we pick only variables corresponding to red points as the left hand side is of the form $S \setminus B$ and in particular at most $k$ of them. Let $R_{\mathrm{sol}}$ be those points and pick the sets $S$ for which $S \subseteq \cup R_{\mathrm{sol}} $. This means that the chosen sets do not cover points outside of $R_{\mathrm{sol}}$ and hence cover at most $k$ red points. Additionally, every blue point $b \in B$ is covered by the definition of the target $T$: the activated \fd $S \setminus B \rightarrow b$ means that we have picked the set S and $b \in B$ participates in an activated \fd. The other direction is analogous. Finally, note that the number of sets created in the above reductions equals the number of \fds, i.e. $|\mathcal{F}|$. Invoking Theorems~\ref{thm:chlamtavc1},~\ref{thm:chlamtavc2} by Chlamtavc et al.~\cite{chlamtavc2023approximating} yields the desired results.
\end{proof}

\section{Integrality Gaps} 
\label{sec:complexity}

\noindent We complement Theorem~\ref{thm:lpapprox} by showing   that the integrality gap of the problem of the LP relaxation is at least $\left(\frac{f-1}{2} \right)^{D-1}$ which is close to what we achieve. Specifically, we prove the following results. 

\begin{theorem} (Integrality gap for D-round-\tcand)
\label{thm:gapfull}
For $f\geq 3$, the integrality gap of the Linear Programming Formulation  for the $D$-round-\tcand problem is  at least
 
\[\left( \frac{n/(D+1)-1}{2}
\right)^{D} = \Omega\left( \frac n D \right)^{D}. \] 

To be more precise, it is at least $\left\lfloor\frac f 2 -1\right\rfloor^{D-1}$, where 
$f:= \mathrm{max}_{i \in V}|{ X\subseteq V: X \rightarrow i \in \mathcal F }|.$  whenever $f \geq 3$. 
The  result holds even when each \fd in $\mathcal F$ has a left hand side with cardinality at most $2$.   For $f=2$ the problem cannot be approximated within a factor of $2-\epsilon$ for any $\epsilon >0$ assuming the Unique Games Conjecture~\cite{khot2005unique}.
\end{theorem}

 It is worth pointing out that 
 the integrality gap is not an artifact of our modelling of the problem, since as we showed in Theorem~\ref{thm:oneround}, there exists no polynomial time algorithm that can achieve an approximation ratio better than $|\mathcal F|^{\frac14 - \delta}$  with $\delta>0$   under the Dense-vs-Random Conjecture~\cite{chlamtavc2017minimizing} even for $D=1$. Before presenting the proof of Theorem~\ref{thm:gapfull}, we first establish Theorem~\ref{thm:gap1}, which is a simpler version of the theorem. This preliminary result helps to elucidate the fundamental concept driving the overall proof. In essence, our proof is analogous to the integrality gap of the standard LP relaxation for the vertex cover, which assigns each variable a value of $\frac{1}{2}$ in a clique of $n$ nodes. This approach results in an integrality gap of $\frac{n-1}{\frac{n}{2}} \to 2$ as $n \to \infty$~\cite{williamson2011design}.  
The adaptation to obtain the more general Theorem~\ref{thm:gapfull}  is straightforward afterwards.

\begin{theorem}
\label{thm:gap1}
The integrality gap of our LP formulation for the $D$-round \tcand problem is at least $2^{ D }$. 
\end{theorem}

\begin{proof} 
We create an instance of the $D$-round \tcand problem as follows. Consider the variable set $V := V^{(0)} \cup V^{(1)} \cup \ldots \cup V^{(D)}$ where $V^{(r)}:= \left\{ (i,r) \right\}_{i \in [5] }$. Note that we use tuples as variable names. We insert \fds in $\mathcal F$ from some pairs of variables in $V^{(r)}$ to variables in $V^{(r+1)}$ as follows for $r=0,\ldots,D-1$. For each variable $(k,r+1)$ in $V^{(r+1)}$, we create two \fds of the form $ (i,r), (j,r) \to (k,r+1) $. Note that $|V^{(r)} \times V^{(r)}| = {5\choose 2} = 10$ pairs of variables in layer $r$. These 10 pairs are then partitioned into 5 sets,
 each containing 2 pairs. Each set corresponds to the left-hand side of an \fd for a specific variable in the next layer, $V^{(r+1)}$. In other words, for each variable in $V^{(r+1)}$, there are two pairs of variables from $V^{(r)}$ that determine it through FDs. Finally, we set $T = V^{(D)}$ as the target set for this instance.   Notice that $T \subseteq (V^{(0)})^+$.


A feasible fractional solution to our LP formulation is to assign the fractional weight $c:=2^{-D+1}$ to every variable in $V^{(0)}$, i.e. set $x_{i,0}^{(0)} = 2^{-D+1}, i \in [5]$. One can verify that $z_{LS}^{(1)} := 1$ for every $LS \rightarrow (i,1)$ and hence we have that 
$x_{i,1}^{(1)}$ can be as large as $\sum_{z_{LS}^{(1)} \rightarrow (i,1)} c = 2c$ for all $i$ by the fact that every $(i,r)$ variable can be inferred from exactly two \fds. Inductively, we have that $x_{i,d}^{(d)}$ can be as large as $2^d c$ and hence in the last layer we have $x_{i,D}^{(D)} = 2^D \cdot c =1$, meaning that every element in the last layer, which is exactly our target set, satisfies the target constraint of our LP.  
 This shows that the total fractional cost is $5 \cdot 2^{-D}$ whereas the integral cost is  $5$, yielding the $2^{D}$ integrality gap. 
 \end{proof}

It is important to note that our proof employs an inductive approach. However, the underlying intuition becomes more apparent when we consider the process starting from the final layer $D$ and moving backwards. In the final layer, each variable is a target variable, so its corresponding LP variable is set to 1 because of the target constraints. Given that each variable can be inferred by two  \fds, an unfavorable yet feasible scenario for the LP would involve equally distributing the weight between these two \fds, and this pattern continues in preceding layers. The complete proof of Theorem~\ref{thm:gapfull} essentially extends this concept by adjusting the number of \fds associated with each variable.

\begin{proof}[Proof of Theorem~\ref{thm:gapfull}]
To obtain the general hardness on the integrality gap, we note that it suffices to set $V^{(r)} := 
\left\{ (i,r)\right\}_{i \in [g]}$ for some parameter $g$. Then we can force each variable in $V^{(0)}, V^{(1)}, \ldots, V^{(D)}$ to participate in ${ g \choose 2} / g \geq \left\lfloor\frac{g-1}{2}\right\rfloor$ \fds. Thus, the base of exponent in the fractional magnification per level is $g-1$ in contrast to $2$ and thus we may obtain both results by setting $g := n/(D+1)$ (for a total of $g \cdot D = n$ nodes) or $g = f = max_{i \in V}|{ X\subseteq V: X \rightarrow i \in \mathcal F }|$ for any  $f \geq 3$. 

For $f=2$, we observe that the problem is as hard as the Vertex Cover problem, so it cannot be approximated within a factor $2-\epsilon$ for any $\epsilon>0$ assuming the Unique Games Conjecture~\cite{khot2008vertex}. To prove the equivalence $G(V,E)$ we create a \emph{non-target} variable $x_v$ for each vertex $v \in V$ and a \emph{target} variable $x_e$ for each edge $e \in E$ adding the FD $x_ux_v \rightarrow x_e$ if $e:=(u,v)$. We first observe that any vertex cover in $G$ corresponds trivially to a set of variables which cover all target variables $\{x_e\}_{e \in E}$. For the other direction, every set $S$ of variables of size $\leq k$ from which $\{x_e\}_{e \in E}$ can be inferred we can create a set $\widetilde{S}$ of size $|\widetilde{S}| \leq |S| \leq k$ by substituting each $x_e \in S, e:=(u,v)$ with $x_u$ (if $x_u$ already in the set) and noting that $\widetilde{S}$ is a vertex cover for $G$.
\end{proof}

\noindent The above instances rule out any hope for designing efficient approximation algorithms for \tcand using linear programming based approaches.   

\section{Conclusion}
\label{sec:concl} 
In this work, we introduce the \tcand problem, a generalization of the well-known minimum candidate key problem~\cite{lipski1977two,lucchesi1978candidate}. The \tcand problem plays an important role in semantic query optimization. We demonstrate that \tcand, in its general form, is a layered set-cover problem, with each layer representing a stage of \fd inference using the given \fds. We formulate the \tcand as an integer program and explore its LP relaxation, both from the perspective of algorithm design and the analysis of integrality gaps. We also examine specific cases of the \tcand problem, such as scenarios where all \fds have at most one attribute on their left-hand side. In cases with one round of inference, the problem aligns with the \rbsetcover, a variant of    \setcover known for its inapproximability.   Our study opens a gateway to a host of compelling challenges, with the development of practical heuristics as a promising direction for future research.

\section{Acknowledgements} 

\noindent CET acknowledges Mihail Kolountzakis for early discussions on the layered set cover problem formulation.

\bibliographystyle{abbrv}
\bibliography{ref}
 
\end{document}